\documentclass[conference,10pt]{IEEEtran}

\usepackage{amsfonts}
\usepackage{amsmath,amssymb}
\usepackage[bookmarks=false]{hyperref}
\usepackage{amsthm}
\usepackage{cite}
\usepackage{float}
\usepackage{slashbox}
\usepackage{nicefrac}
\usepackage[dvipsnames]{xcolor}
\usepackage{enumitem}
\usepackage{tikz}
\usepackage{graphicx}   
\usepackage[makeroom]{cancel}
\usepackage{colortbl}
\usepackage{multirow,geometry,balance}
\theoremstyle{remark}
\newtheorem{theorem}{ {Theorem}}

\newtheorem{definition}{{Definition}}

\newtheorem{remark}{ {Remark}}

\let\olddefinition\definition
\renewcommand{\definition}{\olddefinition\normalfont}
\let\oldtheorem\theorem
\renewcommand{\theorem}{\oldtheorem\normalfont}
\let\oldremark\remark
\renewcommand{\remark}{\oldremark\normalfont}
\usepackage{subcaption}
\usepackage{pgfplots}
\pgfplotsset{compat=newest}

\usetikzlibrary{plotmarks}
\usepackage{grffile}
\usepackage{amsmath}

\usetikzlibrary{shapes}
\pgfplotsset{compat=1.10}
 
\newcommand{\ti}{\textit}

\renewcommand{\IEEEQED}{\IEEEQEDopen}
\definecolor{OliveGreen}{rgb}{0,0.6,0}

\usetikzlibrary{decorations.text}
\usepgfplotslibrary{fillbetween}
\usetikzlibrary{calc, fit}
\usetikzlibrary{arrows.meta}
\usetikzlibrary{patterns}

\newcommand{\basestation}[3]{
	\coordinate (a) at (#1,#2);
	\draw[line width=(#3)*1.5pt,scale=#3] ($(a)+(0, -1.08)$) -- ($(a)+(0, 1)$);
	\draw[line width=(#3)*1.5pt,scale=#3] ($(a)+(0,1)$) .. controls ($(a)+(-0.25,-0.5)$) .. ($(a)+(-0.5,-0.9)$);
	\draw[line width=(#3)*1.5pt,scale=#3] ($(a)+(0,1)$) .. controls ($(a)+(0.25,-0.5)$) .. ($(a)+(0.5,-0.9)$);
	\draw[line width=(#3)*1.5pt,scale=#3] ($(a)+(0, -0.9)$) -- ($(a)+(-0.35, -0.7)$);
	\draw[line width=(#3)*1.5pt,scale=#3] ($(a)+(0, -0.9)$) -- ($(a)+(0.35, -0.7)$);
	\draw[line width=(#3)*0.75pt,scale=#3] ($(a)+(-0.35, -0.65)$) -- ($(a)+(0, -0.5)$);
	\draw[line width=(#3)*0.75pt,scale=#3] ($(a)+(0.35, -0.65)$) -- ($(a)+(0, -0.5)$);
	\draw[line width=(#3)*1pt,scale=#3] ($(a)+(0, -0.6)$) -- ($(a)+(-0.3, -0.45)$);
	\draw[line width=(#3)*1pt,scale=#3] ($(a)+(0, -0.6)$) -- ($(a)+(0.3, -0.45)$);
	\draw[line width=(#3)*0.5pt,scale=#3] ($(a)+(-0.3, -0.45)$) -- ($(a)+(0, -0.32)$);
	\draw[line width=(#3)*0.5pt,scale=#3] ($(a)+(0.3, -0.45)$) -- ($(a)+(0, -0.32)$);
	\draw[line width=(#3)*0.75pt,scale=#3] ($(a)+(0, -0.3)$) -- ($(a)+(-0.22, -0.17)$);
	\draw[line width=(#3)*0.75pt,scale=#3] ($(a)+(0, -0.3)$) -- ($(a)+(0.22, -0.17)$);
	\draw[line width=(#3)*0.5pt,scale=#3] ($(a)+(-0.22, -0.17)$) -- ($(a)+(0, -0.07)$);
	\draw[line width=(#3)*0.5pt,scale=#3] ($(a)+(0.22, -0.17)$) -- ($(a)+(0, -0.07)$);;
	\draw[line width=(#3)*0.75pt,scale=#3] (a) -- ($(a)+(-0.18, 0.11)$);
	\draw[line width=(#3)*0.75pt,scale=#3] (a) -- ($(a)+(0.18, 0.11)$);
	\draw[line width=(#3)*0.5pt,scale=#3] ($(a)+(-0.18, 0.11)$) -- ($(a)+(0,0.2)$);
	\draw[line width=(#3)*0.5pt,scale=#3] ($(a)+(0.18, 0.11)$) -- ($(a)+(0,0.2)$);   
	\draw[line width=(#3)*0.5pt,scale=#3] ($(a)+(0, 0.3)$) -- ($(a)+(-0.1, 0.37)$);
	\draw[line width=(#3)*0.5pt,scale=#3] ($(a)+(0, 0.3)$) -- ($(a)+(0.1, 0.37)$);
	\draw[line width=(#3)*0.25pt,scale=#3] ($(a)+(-0.1, 0.37)$) -- ($(a)+(0, 0.43)$);
	\draw[line width=(#3)*0.25pt,scale=#3] ($(a)+(0.1, 0.37)$) -- ($(a)+(0, 0.43)$);
	\draw[line width=(#3)*0.75pt,scale=#3] ($(a)+(0, 1.2)$) -- ($(a)+(0,1)$);
	\draw[fill=white,scale=#3] ($(a)+(0, 1.2)$) circle (0.05cm);
	\draw[line width=(#3)*1pt,decorate,decoration=expanding waves,decoration={segment length=(#3)*10pt},scale=#3] ($(a)+(0, 1.2)$) -- ($(a)+(0, 2.5)$);
}

\newcommand{\iPhone}[3]{
	\coordinate (a) at (#1,#2);
	\draw [line width=0.25pt,rounded corners=(#3)*1mm,fill=white,scale=(#3)] (a)--($(a)+(0.67,0)$)--($(a)+(0.67,1.381)$)--($(a)+(0,1.381)$)--cycle;
	\draw [color=gray,line width=0.25pt,rounded corners=(#3)*0.8mm,fill=white,scale=(#3)]($(a)+(0.015,0.015)$)--($(a)+(0.655,0.015)$)--($(a)+(0.655,1.366)$)--($(a)+(0.015,1.366)$)--cycle;
	\draw [line width=0.25pt,rounded corners=(#3)*0.04mm,scale=(#3)] ($(a)+(0.2875,1.266)$)--($(a)+(0.3825,1.266)$)--($(a)+(0.3825,1.281)$)--($(a)+(0.2875,1.281)$)--cycle;
	\draw[line width=0.25pt,scale=#3] ($(a)+(0.335,0.09)$) circle (0.055cm);
	\draw[line width=0.25pt,scale=#3] ($(a)+(0.335,0.09)$) circle (0.044cm);
	\draw[line width=0.25pt,scale=#3] ($(a)+(0.2275,1.2735)$) circle (0.015cm);
	\draw[line width=0.25pt,scale=#3] ($(a)+(0.335,1.32)$) circle (0.01cm);
	\draw [fill={rgb:black,1;white,4},line width=0.25pt,scale=(#3)] ($(a)+(0.042475,0.170195)$)--($(a)+(0.042475,0.170195)+(0.58505,0.0)$)--($(a)+(0.042475,0.170195)+(0.58505,1.04061)$)--($(a)+(0.042475,0.170195)+(0.0,1.04061)$)--cycle;
}

\newcommand{\SymModChanged}[0]{
	\begin{tikzpicture}[scale=0.68,
	roundnode/.style={circle, draw=green!60, fill=green!5, very thick, minimum size=7mm},
	squarednode/.style={rectangle, draw=red!60, fill=red!5, very thick, minimum size=5mm},
	]
	\basestation{0.15}{2.6+2}{0.65};
	\node[scale=0.75] at (0.15,2.6+1.05) {BS};

	\basestation{2+0.75}{0.6+1.9}{0.425};
	\node (ch) at (2+1.25, 0.4+1.9) {\includegraphics[scale=0.8]{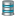}};
	\node[scale=0.75] at (2+0.85,-0.15+1.9) {RN$_1$};
	\iPhone{3.75+2.7}{2.6+1.7}{0.5};
	\node[scale=0.75] at (3.75+2.9,2.2+1.8) {UE$_1$};
	
	\draw[->, black, thick] (0.7, 2.6+2) -- (3.3+3-0.05, 2.6+2) node[pos=0.5,sloped,above] {$g_{1}$};
	\draw[->, thick, blue] (0.7, 2.6+2) -- (1.5+1-0.05, 0.6+2) node[pos=0.5,sloped,above] {$f_1$};
	\draw[->, magenta, thick] (2.45+1+0.1, 0.6+2) -- (3.25+3-0.05, 2.5+2) node[pos=0.5,sloped,below] {$h_{11}$};

	\basestation{2+0.75}{0}{0.425};
	\node (ch) at (2+1.25, -0.2) {\includegraphics[scale=0.8]{content/database-icon16px.png}};
	\node[scale=0.75] at (2+0.85,-0.75) {RN$_M$};
	\iPhone{3.75+2.7}{-0.3}{0.5};
	\node[scale=0.75] at (3.75+2.9,-0.6) {UE$_K$};
	
	\draw[->, black, thick] (0.7, 2.6+2) to [out=-10, in=-230] (3.25+3-0.05, 0);
	\node at (3.475, 3.65) {$g_{K}$};
	\draw[->, thick, blue] (0.7, 2.6+2) -- (1.4+1-0.05, 0) node[pos=0.5,sloped,above] {$f_M$};
	\draw[->, magenta, thick] (2.45+1+0.1, 0.6+2) -- (3.25+3-0.05, 0) node[pos=0.55,sloped,below] {$h_{K1}$};
	\draw[->, magenta, thick] (2.5+1+0.1, 0) -- (3.25+3-0.05, 2.5+2) node[pos=0.5,sloped,below] {$h_{1M}$};
	\draw[->, magenta, thick] (2.5+1+0.1, 0) -- (3.25+3-0.05, 0) node[pos=0.5,sloped,below] {$h_{KM}$};
	
	\path (3.75+2.9, 2.6+2) -- (3.75+2.9, 0) node [black, scale=3.0, midway, sloped] {$\dots$};
	\path (2+0.825, 0.24+1.8) -- (2+0.825, 0.45) node [black, scale=1, midway, sloped] {$\dots$};
	\end{tikzpicture}
}

\newcommand{\SymModSchemePhaseOne}[0]{
	\begin{tikzpicture}[
	roundnode/.style={circle, draw=green!60, fill=green!5, very thick, minimum size=7mm},
	squarednode/.style={rectangle, draw=red!60, fill=red!5, very thick, minimum size=5mm},scale=0.65
	]
	
	\basestation{-0.1}{2.6+2}{0.65};
	\node[scale=0.75] at (-0.1,2.6+1.05) {BS};

	\basestation{1+2+0.75}{1+1.9}{0.425};
	\node (ch) at (1+2+1.25, 0.8+1.9) {\includegraphics[scale=0.7]{content/database-icon16px.png}};
	\node[scale=0.75] at (1+2+1.0,1.9+0.25) {RN$_1$};
	\node[rectangle,draw,rounded corners=5pt,minimum width=1cm, minimum height=1.4cm, very thick, blue] at (1+2+1,1+1.9) {};
	
	\basestation{1+2+0.75}{-1.3+1.9}{0.425};
	\node (ch) at (1+2+1.25, -1.5+1.9) {\includegraphics[scale=0.7]{content/database-icon16px.png}};
	\node[scale=0.75] at (1+2+1.0,-1.3+1.9-0.75) {RN$_2$};
	\node[rectangle,draw,rounded corners=5pt,minimum width=1cm, minimum height=1.4cm, very thick, blue] at (1+2+1,-1.3+1.9) {};
	
	\basestation{1+2+0.75}{-3.6+1.9}{0.425};
	\node (ch) at (1+2+1.25, -3.8+1.9) {\includegraphics[scale=0.7]{content/database-icon16px.png}};
	\node[scale=0.75] at (1+2+1.0,-3.6+1.9-0.75) {RN$_3$};
	\node[rectangle,draw,rounded corners=5pt,minimum width=1cm, minimum height=1.4cm, very thick, blue] at (1+2+1,-3.6+1.9) {};
	\node[rectangle,draw,rounded corners=5pt,minimum width=1.2cm, minimum height=1.6cm, very thick, red] at (1+2+1,-3.6+1.9) {};
	
	\basestation{1+2+0.75}{-5.9+1.9}{0.425};
	\node (ch) at (1+2+1.25, -6.1+1.9) {\includegraphics[scale=0.7]{content/database-icon16px.png}};
	\node[scale=0.75] at (1+2+1.0,-5.9+1.9-0.75) {RN$_4$};
	\node[rectangle,draw,rounded corners=5pt,minimum width=1cm, minimum height=1.4cm, very thick, red] at (1+2+1,-5.9+1.9) {};
	
	\iPhone{1+3.75+2.7+0.25+1}{2.6+1.7-0.1}{0.5};
	\node[scale=0.75] at (1+3.75+2.9+0.25+1,2.2+1.8-0.1) {UE$_1$};
	\iPhone{1+3.75+2.7+0.25+1}{0.6+1.9}{0.5};
	\node[scale=0.75] at (1+3.75+2.9+0.25+1,0.2+1.9) {UE$_2$};
	
	\draw[dashed,blue, very thick] (0.25+1.3,-1.35+1.9) ellipse (0.85cm and 3cm);
	\node[scale=0.85, align=center, blue] at (0.25+1.3,-1.35+1.9) {SISO \\[-0.75ex] Broad- \\[-0.75ex] cast:\\[-0.75ex] MAN\\[-0.75ex] scheme};
	
	\draw[dashed,blue, thick] (1.5+5.05,0.6+1.9) ellipse (0.85cm and 3cm);
	\node[scale=0.85, align=center, blue] at (1.5+5.05,0.6+1.9) {MISO \\ [-0.75ex]Broad- \\[-0.75ex] cast:\\[-0.75ex] ZF\\[-0.75ex] Inter-\\[-0.75ex]ference};
	
	\draw[dashed,red, very thick] (1.5+5.05,-4.95+1.9) ellipse (0.85cm and 2cm);
	\node[scale=0.85, align=center, red] at (1.5+5.05,-4.65+1.9) {MISO \\ [-0.75ex]Broad- \\[-0.75ex] cast:\\[-0.75ex] UE\\[-0.75ex] Desired\\[-0.75ex]Symbols};
	
	\draw[-, thick, blue] (0.35, 2.6+2) -- (0.95, 0.8+2);
	\draw[->, thick, blue] (0.475+1.675, 0.8+2) -- (0.95+2.35, 0.8+2);
	\draw[-, thick, blue] (0.35, 2.6+2) -- (0.7, -1.4+2);
	\draw[->, thick, blue] (2.425, -1.4+2) -- (0.95+2.35, -1.4+2);
	\draw[-, thick, blue] (0.35, 2.6+2) -- (0.5, -3+2) -- (0.85, -3+2);
	\draw[->, thick, blue] (2.25, -3+2) -- (0.75+2.35, -3.8+2);
	\draw[-, thick, gray!55] (0.35, 2.6+2) -- (0.35, -3.8+2) node[pos=0.5,sloped,below,scale=0.85] {Inactive} -- (0.99, -3.8+2);
	\draw[->, thick, gray!55] (2.1, -3.8+2) -- (0.85+2.35, -6+2);
	
	\draw[-, thick, blue] (1.05+3.65,0.8+2) -- (1.1+4.575, 0.8+2);
	\draw[->, thick, blue] (1.8+5.565,0.8+2) -- (6.575+2, 2.6+2);
	\draw[->, thick, blue] (1.8+5.565,0.8+2) -- (6.575+2, 0.8+2);
	\draw[-, thick, blue] (1.05+3.65,-1.4+2) -- (1.2+4.675, -1.4+2);
	\draw[->, thick, blue] (1.8+5.435,-1.4+2) -- (6.575+2, 2.6+2);
	\draw[->, thick, blue] (1.8+5.435,-1.4+2) -- (6.575+2, 0.8+2);
	\draw[-, thick, blue] (1.1+3.775,-3.8+2) -- (1.3+4.725, -2+2);
	\draw[-, thick, red] (1.1+3.775,-3.8+2) -- (1.3+4.625, -3.8+2);
	\draw[->, thick, red] (1.75+5.435,-3.8+2) -- (6.575+2, 2.6+2);
	\draw[->, thick, red] ((1.75+5.435,-3.8+2) -- (6.575+2, 0.8+2);
	\draw[-, thick, red] (1.05+3.65,-6+2) -- (1.2+4.625, -6+2);
	\draw[->, thick, red] (1.75+5.535,-6+2) -- (6.575+2, 2.6+2);
	\draw[->, thick, red] (1.75+5.535,-6+2) -- (6.575+2, 0.8+2);
	
	\draw[-, thick, blue] (0.356, 2.6+2) -- (4.7+1.25, 2.6+2);
	\draw[->, thick, blue] (1.725+5.425,2.6+2) -- (6.575+2, 2.6+2);
	\draw[->, thick, blue] (1.725+5.425,2.6+2) -- (6.575+2, 0.8+2);	
	
	\end{tikzpicture}
}

\newcommand{\SymModSchemePhaseTwo}[0]{
	\begin{tikzpicture}[
	roundnode/.style={circle, draw=green!60, fill=green!5, very thick, minimum size=7mm},
	squarednode/.style={rectangle, draw=red!60, fill=red!5, very thick, minimum size=5mm},
	scale=0.65]

	\basestation{-0.1}{2.6+2}{0.65};
	\node[scale=0.75] at (-0.1,2.6+1.05) {BS};
	
	\basestation{1+2+0.75}{1+1.9}{0.425};
	\node (ch) at (1+2+1.25, 0.8+1.9) {\includegraphics[scale=0.7]{content/database-icon16px.png}};
	\node[scale=0.75] at (1+2+1.0,1.9+0.25) {RN$_1$};
	\node[rectangle,draw,rounded corners=5pt,minimum width=1cm, minimum height=1.4cm, very thick, blue] at (1+2+1,1+1.9) {};
	
	\basestation{1+2+0.75}{-1.3+1.9}{0.425};
	\node (ch) at (1+2+1.25, -1.5+1.9) {\includegraphics[scale=0.7]{content/database-icon16px.png}};
	\node[scale=0.75] at (1+2+1.0,-1.3+1.9-0.75) {RN$_2$};
	\node[rectangle,draw,rounded corners=5pt,minimum width=1cm, minimum height=1.4cm, very thick, blue] at (1+2+1,-1.3+1.9) {};

	\basestation{1+2+0.75}{-3.6+1.9}{0.425};
	\node (ch) at (1+2+1.25, -3.8+1.9) {\includegraphics[scale=0.7]{content/database-icon16px.png}};
	\node[scale=0.75] at (1+2+1.0,-3.6+1.9-0.75) {RN$_3$};
	\node[rectangle,draw,rounded corners=5pt,minimum width=1cm, minimum height=1.4cm, very thick, gray!55] at (1+2+1,-3.6+1.9) {};
	
	\basestation{1+2+0.75}{-5.9+1.9}{0.425};
	\node (ch) at (1+2+1.25, -6.1+1.9) {\includegraphics[scale=0.7]{content/database-icon16px.png}};
	\node[scale=0.75] at (1+2+1.0,-5.9+1.9-0.75) {RN$_4$};
	\node[rectangle,draw,rounded corners=5pt,minimum width=1cm, minimum height=1.4cm, very thick, gray!55] at (1+2+1,-5.9+1.9) {};

	\node[scale=0.9, gray!55, align=center] at (1.5,-5.05+1.9) {Inactive RNs \\[-0.75ex] (RN$_3$, RN$_4$)};
	
	\iPhone{1+3.75+2.7+0.25+1}{2.6+1.7-0.1}{0.5};
	\node[scale=0.75] at (1+3.75+2.9+0.25+1,2.2+1.8-0.1) {UE$_1$};
	\iPhone{1+3.75+2.7+0.25+1}{0.6+1.9}{0.5};
	\node[scale=0.75] at (1+3.75+2.9+0.25+1,0.2+1.9) {UE$_2$};
	
	\draw[dashed,red, thick] (1.5+5.05,0.6+1.9) ellipse (0.85cm and 3cm);
	\node[scale=0.85, align=center, red] at (1.5+5.05,0.6+1.9) {MISO \\ [-0.75ex]Broad- \\[-0.75ex] cast:\\[-0.75ex] UE\\[-0.75ex] Desired\\[-0.75ex]Symbols};

	\draw[-, thick, red] (1.05+3.65,0.8+2) -- (1.1+4.575, 0.8+2);
	\draw[->, thick, red] (1.8+5.565,0.8+2) -- (6.575+2, 2.6+2);
	\draw[->, thick, red] (1.8+5.565,0.8+2) -- (6.575+2, 0.8+2);
	\draw[-, thick, red] (1.05+3.65,-1.4+2) -- (1.2+4.675, -1.4+2);
	\draw[->, thick, red] (1.8+5.435,-1.4+2) -- (6.575+2, 2.6+2);
	\draw[->, thick, red] (1.8+5.435,-1.4+2) -- (6.575+2, 0.8+2);
	
	\draw[-, thick, red] (0.356, 2.6+2) -- (4.7+1.25, 2.6+2);
	\draw[->, thick, red] (1.725+5.425,2.6+2) -- (6.575+2, 2.6+2);
	\draw[->, thick, red] (1.725+5.425,2.6+2) -- (6.575+2, 0.8+2);	
	
	\end{tikzpicture}
}

\newcommand{\PwrLevelPOne}[0]{
	\draw [<-,thick] (0,3.75) node[scale=0.9] (yaxis) [above] {$Power Level$} |- (2.15,0) node[scale=0.9] (xaxis) [right] {};
		
	\draw[very thick, dashed] (0.0, 0.3) -- (0.2, 0.3);
	\draw[very thick] (0.2, 0.3) -- (2.15, 0.3);
	\node[scale=0.65] at (1.125, 0.15) {Noise Floor};
		
	\draw [thick, fill=blue!10] (0.5,0.3) rectangle (1.5,3.25) node[align=center, pos=.5, rotate=90, scale=0.8, blue] {Desired RN \\ Symbols $\mathcal{S}_R$};
		
	\node[scale=0.8] at (-0.5, 0.3) {$P^{0}$};
	\node[scale=0.8] at (-0.5, 3.25) {$P$};
	\draw[dashed] (0.0, 3.25) -- (0.5, 3.25);
	\draw[dashed] (1.6, 3.25) -- (2.15, 3.25);
		
	\node at (1.05, -1) {\textbf{BS}};
	
	\draw [<-,thick] (4.5+0,3.75) node[scale=0.9] (yaxis) [above] {$Power Level$} |- (4.5+3.75,0) node[scale=0.9] (xaxis) [right] {};
		
	\draw[very thick, dashed] (4.5+0.0, 0.3) -- (4.5+0.2, 0.3);
	\draw[very thick] (4.5+0.2, 0.3) -- (4.5+3.75, 0.3);
	\node[scale=0.65] at (4.5+2, 0.15) {Noise Floor};
		
	\draw [thick, fill=blue!10] (4.5+0.5,0.3) rectangle (4.5+1.5,3.25) node[align=center, pos=.5, rotate=90, scale=0.80, blue] {Desired RN$_m$ \\ Symbol $(m\notin\mathcal{S}_R)$};
	\node[scale=0.9] at (4.5+1, -0.5) {Rx};
		
	\draw [thick, fill=blue!10] (4.5+2.5,0.3) rectangle (4.5+3.5,3.25) node[align=center, pos=.5, rotate=90, scale=0.80, red] {Desired UE \\ Symbols};
	\node[scale=0.9] at (4.5+3, -0.5) {Tx};
		
	\node[scale=0.8] at (4.5-0.5, 0.3) {$P^{0}$};
	\node[scale=0.8] at (4.5-0.5, 3.25) {$P$};
	\draw[dashed] (4.5+0.0, 3.25) -- (4.5+0.5, 3.25);
	\draw[dashed] (4.5+1.6, 3.25) -- (4.5+2.4, 3.25);
	\draw[dashed] (4.5+3.6, 3.25) -- (4.5+3.75, 3.25);	
		
	\node at (4.5+2, -1) {\textbf{RN}$_\mathbf{m}$};
	
	\draw [<-,thick] (9.75+0,3.75) node[scale=0.9] (yaxis) [above] {$Power Level$} |- (9.75+4.25,0) node[scale=0.9] (xaxis) [right] {};
	
	\draw[very thick, dashed] (9.75+0.0, 0.3) -- (9.75+0.2, 0.3);
	\draw[very thick] (9.75+0.2, 0.3) -- (9.75+4.25, 0.3);
	\node[scale=0.65] at (9.75+2.25, 0.15) {Noise Floor};
	
	\draw [thick, fill=blue!10] (9.75+0.5,0.3) rectangle (9.75+1.5,3.25) node[align=center, pos=.5, rotate=90, scale=0.80, red] {Desired UE$_k$ \\ Symbol $(k\in\mathcal{S}_U)$};
	\draw [thick, fill=blue!10] (9.75+1.75,0.3) rectangle (9.75+2.75,2.4) node[align=center, pos=.5, rotate=90, scale=0.80, blue] {Interference \\ RN Symbols};
	\draw [thick, fill=blue!10] (9.75+3,0.3) rectangle (9.75+4,2.4) node[align=center, pos=.5, rotate=90, scale=0.80, red] {Interference \\ UE Symbols};
	
	\draw[dashed] (9.75+0.0, 2.4) -- (9.75+0.5, 2.4);
	\draw[dashed] (9.75+1.55, 2.4) -- (9.75+1.75, 2.4);
	\draw[dashed] (9.75+2.8, 2.4) -- (9.75+3, 2.4);
	\draw[dashed] (9.75+4.1, 2.4) -- (9.75+4.25, 2.4);
	\draw[dashed] (9.75+0.0, 3.25) -- (9.75+0.5, 3.25);
	\draw[dashed] (9.75+1.6, 3.25) -- (9.75+4.25, 3.25);
	
	\node[scale=0.8] at (9.75+-0.5, 0.3) {$P^{0}$};
	\node[scale=0.8] at (9.75+-0.5, 2.4) {$P^{1-\alpha}$};
	\node[scale=0.8] at (9.75+-0.5, 3.25) {$P$};
	
	\node at (9.75+2.25, -1) {\textbf{UE}$_\mathbf{k}$};
	
	\draw[thick, <->] (9.75+2.25, 2.4) -- (9.75+2.25, 3.25) node[pos=0.5, right, scale=0.8] {$P^{\alpha}$};
	
}

\newcommand{\PwrLevelPTwo}[0]{
	\draw [<-,thick] (0,3.75) node[scale=0.9] (yaxis) [above] {$Power Level$} |- (3,0) node[scale=0.9] (xaxis) [right] {};
	
	\draw[very thick, dashed] (0.0, 0.3) -- (0.2, 0.3);
	\draw[very thick] (0.2, 0.3) -- (3, 0.3);
	\node[scale=0.65] at (1.5, 0.15) {Noise Floor};
	
	\draw [thick, fill=blue!10] (0.5,0.3) rectangle (1.5,3.25) node[align=center, pos=.5, rotate=90, scale=0.8, red] {\textcolor{black}{\emph{Common}} Symbol \\ desired by UE$_{\bar{k}}$};
	\draw [thick, fill=blue!10] (1.75,0.3) rectangle (2.75,2.5) node[align=center, pos=.5, rotate=90, scale=0.8, red] {\textcolor{magenta}{\emph{Private}} UE \\Symbols};
	
	\node[scale=0.8] at (-0.5, 0.3) {$P^{0}$};
	\node[scale=0.8] at (-0.5, 2.5) {$P^{\alpha}$};
	\node[scale=0.8] at (-0.5, 3.25) {$P$};
	\draw[dashed] (0.0, 3.25) -- (0.5, 3.25);
	\draw[dashed] (1.6, 3.25) -- (3, 3.25);
	\draw[dashed] (0.0, 2.5) -- (0.5, 2.5);
	\draw[dashed] (1.55, 2.5) -- (1.75, 2.5);
	\draw[dashed] (2.8, 2.5) -- (3, 2.5);
	
	\node at (1.5, -1) {\textbf{BS}};
	
	\draw [<-,thick] (4.5+0,3.75) node[scale=0.9] (yaxis) [above] {$Power Level$} |- (4.5+3.75,0) node[scale=0.9] (xaxis) [right] {};
	
	\draw[very thick, dashed] (4.5+0.0, 0.3) -- (4.5+0.2, 0.3);
	\draw[very thick] (4.5+0.2, 0.3) -- (4.5+3.75, 0.3);
	\node[scale=0.65] at (4.5+2, 0.15) {Noise Floor};

	\draw [thick, fill=blue!10] (4.5+0.5,0.3) rectangle (4.5+1.5,3.25) node[align=center, pos=.5, rotate=90, scale=0.80, red] {Interference \\ UE Symbols};
	\node[scale=0.9] at (4.5+1, -0.5) {Rx};
	
	\draw [thick, fill=blue!10] (4.5+2.5,0.3) rectangle (4.5+3.5,3.25) node[align=center, pos=.5, rotate=90, scale=0.80, red] {Desired UE \\ Symbols};
	\node[scale=0.9] at (4.5+3, -0.5) {Tx};
	
	\node[scale=0.8] at (4.5-0.5, 0.3) {$P^{0}$};
	\node[scale=0.8] at (4.5-0.5, 3.25) {$P$};
	\draw[dashed] (4.5+0.0, 3.25) -- (4.5+0.5, 3.25);
	\draw[dashed] (4.5+1.6, 3.25) -- (4.5+2.4, 3.25);
	\draw[dashed] (4.5+3.6, 3.25) -- (4.5+3.75, 3.25);	
	
	\node at (4.5+2, -1) {\textbf{RN}$_\mathbf{m}$};
	
	\draw [<-,thick] (9.75+0,3.75) node[scale=0.9] (yaxis) [above] {$Power Level$} |- (9.75+3,0) node[scale=0.9] (xaxis) [right] {};
	
	\draw[very thick, dashed] (9.75+0.0, 0.3) -- (9.75+0.2, 0.3);
	\draw[very thick] (9.75+0.2, 0.3) -- (9.75+3, 0.3);
	\node[scale=0.65] at (9.75+1.5, 0.15) {Noise Floor};
	
	\draw [thick, fill=blue!10] (9.75+0.5,0.3) rectangle (9.75+1.5,3.25) node[align=center, pos=.5, rotate=90, scale=0.80, red] {\textcolor{black}{\emph{Common}} Symbol \\ desired by UE$_{\bar{k}}$};
	\draw [thick, fill=blue!10] (9.75+1.75,0.3) rectangle (9.75+2.75,2.5) node[align=center, pos=.5, rotate=90, scale=0.80, red] {\textcolor{magenta}{Private} UE$_k$ \\ Symbol};
	
	\draw[dashed] (9.75+0.0, 2.5) -- (9.75+0.5, 2.5);
	\draw[dashed] (9.75+1.55, 2.5) -- (9.75+1.75, 2.5);
	\draw[dashed] (9.75+2.8, 2.5) -- (9.75+3, 2.5);
	\draw[dashed] (9.75+0.0, 3.25) -- (9.75+0.5, 3.25);
	\draw[dashed] (9.75+1.6, 3.25) -- (9.75+3, 3.25);
	
	\node[scale=0.8] at (9.75-0.5, 0.3) {$P^{0}$};
	\node[scale=0.8] at (9.75-0.5, 2.4) {$P^{\alpha}$};
	\node[scale=0.8] at (9.75-0.5, 3.25) {$P$};
	
	\node at (9.75+1.5, -1) {\textbf{UE}$_\mathbf{k}$};

	\draw[thick, <->] (9.75+2.25, 2.5) -- (9.75+2.25, 3.25) node[pos=0.5, right, scale=0.8] {$P^{1-\alpha}$};
	
}

\newcommand{\PlotMMuRegionsTwo}[0]{
	\begin{axis}[
		axis x line=bottom, axis y line=left,
		ymin=1, ymax=8.3, ytick={1}, extra y ticks={2,5}, ylabel={$M$},
		xmin=0, xmax=1.05, xtick={0,0.2, 0.4, 1}, xlabel=$\mu$,
		xticklabels={\footnotesize $0$, \scriptsize $\frac{1}{(2K+1)}=\nicefrac{1}{5}$, \scriptsize $\quad\quad\quad\frac{K}{(2K+1)}=\nicefrac{2}{5}$,\scriptsize $1$},
		yticklabels={\footnotesize $1$}, extra y tick labels={\footnotesize $K$,\footnotesize $2K+1$}, extra y tick style={y tick label style={rotate=90}},
		]
		
		\addplot[name path=f1, very thick, dashed, cyan, domain=0:1, smooth]{2} node[above, pos=0.35] {$K$};
		
		\path[name path=axis] (axis cs:1,2) -- (axis cs:1,0) -- (axis cs:0,0);
		\addplot[pattern=north west lines, pattern color=OliveGreen] fill between[of=axis and f1];
		
		\addplot+[name path=f2, no markers, decoration={
			text={}, raise=1ex,
			text align={center}
		}, postaction={decorate}, dashed, very thick, blue, domain=0.4:1, smooth]{2/x} node[above, pos=0.95, scale=1.2] {$\frac{K}{\mu}$};
		
		\addplot[name path=f3, very thick, dashed, OliveGreen, domain=0.4:0.4375, smooth]{1/(1-2*x)} node[below, pos=0.65, scale=1.05] {$\qquad\frac{1}{1-2\mu}$};
		
		\path[name path=yc] (axis cs:0.4375,8) -- (axis cs:1,8);
		
		\path[name path=xc] (axis cs:1,2) -- (axis cs:1,8);
		
		\addplot[pattern=north west lines, pattern color=magenta!55] fill between[of=f2 and yc];
		
		\addplot+[name path=f4, no markers, decoration={
			text={}, raise=1ex,
			text align={center}
		}, postaction={decorate}, dashed, very thick, blue, domain=0.25:0.4, smooth]{2/x};
		
		\path[name path=yc1] (axis cs:0.25,8) -- (axis cs:0.4375,8);
		
		\addplot[pattern=north west lines, pattern color=black!50] fill between[of=f4 and yc1];
		
		\addplot+[name path=f5, no markers, decoration={
			text={}, raise=1ex,
			text align={center}
		}, postaction={decorate}, dashed, very thick, magenta, domain=0.044:0.4, smooth]{2/(2*(1-x))+sqrt(4/(4*(1-x)^(2))+2/(x*(1-x)))} node[below, pos=0.925, scale=0.8] {$K=\mu M\delta_{\text{MAN}}(\mu)$};
		
		\addplot +[mark=none, dashed, very thick, violet] coordinates {(0.5, 4) (0.5, 6)} node[below, pos=0.375, scale=0.8] {$\qquad\quad\mu=\frac{1}{2}$}; 
		
		\addplot +[mark=none, dashed, very thick, violet] coordinates {(0.5, 7) (0.5, 8)};
		
		\addplot[pattern=north west lines, pattern color=cyan!90] fill between[of=f4 and f5];
		
		\foreach \M in {1,2,...,8}{
			\foreach \k in {0,..., \M}{ 
				\addplot[red, mark=otimes*, only marks] coordinates {({\k/\M},{\M})};
			}
		}
		
		\node[anchor=west, scale=0.8, align=left] at (axis cs: 0.525,5.4) {\textbf{Region A}\\[-0.95ex] (Constant NDT-plateau)};
		\node[anchor=west, rotate=-80, scale=0.80] at (axis cs: 0.335,7.65) {\textbf{Region B}};
		\node[scale=0.80] at (axis cs: 0.475,1.5) {\textbf{Region C} ($\#$ RN bottleneck)};
		\node[rotate=-60, scale=0.80, align=center] at (axis cs: 0.145,3.5) {\textbf{Region D} \\[-0.75ex]($\#$ UE bottleneck) };
		\node[rotate=-65, scale=0.80] at (axis cs: 0.225,6.5) {\textbf{Region E}};
		
	\end{axis}
	\draw[->, thick, black] (0.2,-0.6) -- (0.6,1.85);
	\draw[->, thick, black] (0.2,-0.6) -- (1.20,0.4);
	\node[align=center, scale=0.85] at (0.35,-1.2) {Interference \\[-0.75ex](channel)\\[-0.75ex]-limited\\[-0.75ex](UE side)};
	\draw[->, thick, black] (1.6,6.25) -- (1.3,4.9);
	\draw[->, thick, black] (1.6,6.25) -- (2.2,5.15);
	\node[align=center, scale=0.85] at (1.6,6.55) {Broadcast (channel)-limited\\[-0.75ex](RN side)};
	\draw[->, thick, blue] (6.0,-0.4) -- (6.0,0.9);
	\node[align=center, scale=0.85, color=blue] at (5.8,-0.75) {RN standalone \\[-0.75ex]frontier};
}

\newcommand{\alert}[1]{\textcolor{black}{#1}}
\newcommand{\Alert}[1]{\textcolor{black}{#1}}

\begin{document}

\title{Delivery Time Minimization in Cache-Assisted Broadcast-Relay Wireless Networks with Imperfect CSI}

\author{\IEEEauthorblockN{Jaber Kakar$^{*}$, Anas Chaaban$^{\dagger}$, Aydin Sezgin$^{*}$ and Arogyaswami Paulraj$^{\ddagger}$}
\IEEEauthorblockA{$^{*}$Institute of Digital Communication Systems,
Ruhr-Universit{\"a}t Bochum, Germany\\$^{\dagger}$School of Engineering, University of British Columbia, Kelowna, Canada\\$^{\ddagger}$Information Systems Laboratory, Stanford University, CA, USA\\
Email: \{jaber.kakar, aydin.sezgin\}@rub.de, anas.chaaban@ubc.ca, apaulraj@stanford.edu
}}

\maketitle

\begin{abstract}
An emerging trend of next generation communication systems is to provide network edges with additional capabilities such as storage resources in the form of caches to reduce file delivery latency. To investigate the impact of this technique on latency, we study the delivery time of a cache-aided broadcast-relay wireless network consisting of one central base station, $M$ cache-equipped transceivers and $K$ receivers under finite precision channel state information (CSI). We use the normalized delivery time (NDT) to capture the worst-case per-bit latency in a file delivery. Lower and upper bounds on the NDT are derived to understand the influence of $K,M$, cache capacity and channel quality on the NDT. \alert{In particular, regimes of NDT-optimality are identified and discussed.}                
\end{abstract}

% no keywords

\IEEEpeerreviewmaketitle

% no \IEEEPARstart
\section{Introduction}
\label{sec:intro}

Wireless traffic is drastically increasing, particularly due to on-demand video streaming. A promising solution to tackle this problem is \emph{caching}, i.e., storing popular files (video e.g.) in mobile users' local caches and/or \emph{edge nodes} (e.g, base stations (BS) or relays) disseminated in the network coverage area. The local availability of requested user content in the caches, also referred to as \emph{cache hits}, results in reduced backhaul traffic and low file \emph{delivery time}\footnote{In this context, delivery time refers to the timing overhead required to satisfy all file demands of requesting nodes in the network.}.     

In this paper, we consider a broadcast-relay channel (BRC) with cache-assisted relay nodes (RN). As illustrated in Fig. \ref{fig:HetNet}, the network consists of $K$ mobile users (UE$_1$ through UE$_K$), $M$ RNs (RN$_1$ through RN$_M$) and a BS. With the exception of the BS, all remaining nodes request files that ideally ought to be delivered under the \emph{lowest delivery time} possible. To this end, the schemes for RN \emph{cache prefetching/placement} and BS-RN file \emph{delivery} have to be designed and optimized \emph{jointly}. Such joint process involves an optimal balance in delivery times with respect to the file delivery (i) to the RNs from the BS through the broadcast channel (BC) and (ii) to the UEs through the BS-RN interference channel (IC). Regarding (i) and (ii), the RN cache prefetching has to be chosen to facilitate multicasting opportunities %on the BS-RN broadcast channel 
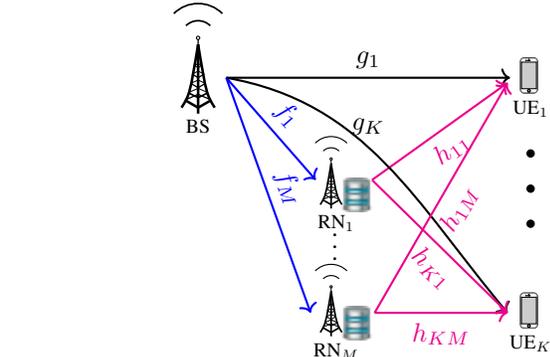
\begin{figure}[h]
	\begin{center}
		%\vspace{-0.5em}
		\begin{tikzpicture}[scale=1]
		\SymModChanged
		\end{tikzpicture}
		\vspace{-3em}
		\caption{\small{A transceiver cache-aided BRC consisting of one BS, $M$ RNs and $K$ UEs. \Alert{These nodes are connected through the wireless links $f_i,g_k$ and $h_{ij}$, $i=1,\ldots,M$, $j=1,\ldots,K$. Each RN is equipped with a finite size cache.}}}	
		\label{fig:HetNet}
	\end{center}
\end{figure} on the BC with respect to files requested by the RNs, in addition to interference coordination techniques (e.g., interference alignment, zero-forcing) through BS-RN cooperation on %the interference channel 
the IC with the BS and RNs as transmitters and the UEs as receivers. In other words, the RNs caches have a dual purpose, i.e., they represent transmitter and receiver caches with respect to the files UEs and RNs request. Thus, in short, we call this network a \emph{transceiver} cache-aided BRC. Such type of network is of importance from an \emph{online cache update} perspective in which RNs refresh their cached contents while simultaneously satisfying the UEs file demands in collaboration with the BS.   

In the existing literature, the effect of caching on the delivery time has predominantly been studied for \emph{interference-limited} networks. \alert{In particular, the study of an error-free BC with $K$ single-antenna receivers endowed with caches showed that the delivery time scales} as $\frac{K(1-\mu)}{1+\mu K}$, where $\mu$ denotes the per-user cache memory size normalized by the entire set of files \cite{Maddah-Ali2}. \alert{Their} work reveals that in addition to the local caching gain of $(1-\mu)$ %(since some content is already available in the receivers' cache which can be served directly without transmission from the local cache memory) 
resulting from the availability of some content in receivers' caches, an extra global caching gain of $\frac{1}{1+\mu K}$ is also attained. The global caching gain originates from multicasting opportunities in the delivery phase that emerge from an appropriate choice in the cache placement. % under a certain pattern. 
More recently, various related settings of \cite{Maddah-Ali2} have been studied. This includes, amongst other, device-to-device caching in D2D networks \cite{Ji16}, IC with either transmitter caches only (with/without cloud processing) \cite{Maddah_Ali,KakarICC} or with caches at both transmitter and receiver under one-shot linear delivery schemes \cite{Naderializadeh}. %and interference alignment schemes \cite{Xu17}. 
The first delivery time characterizations of transceiver-cache aided BRCs for special cases of $K$ and $M$ for equally strong wireless links when $K+M\leq 4$ and non-equally strong wireless links when $(K,M)=(1,1)$ are established in \cite{conference214,KakarICC2018} \alert{under perfect-quality CSI}. 

In the following sections to come, we first present the system model followed by a general information-theoretic lower and upper bound on the NDT for any number of RNs and UEs. The upper bound integrates multicasting and distributed zero-forcing schemes on BC and IC of the BRC. Through comparison of these two bounds, we identify and discuss regions of NDT-optimality in terms of $K,M,\alpha$ and the fractional cache size $\mu$. 

\textbf{Notation:} %We use $\mathcal{N}(\boldsymbol{A})$ to denote the (right) nullspace of matrix $\boldsymbol{A}$. Further, $\boldsymbol{1}_{A}(w)$ denotes the indicator function of $A$ and equals $1$ if $w\in A$ and $0$ otherwise. 
For any two integers $a$ and $b$ with $a\leq b$, we define $[a:b]\triangleq\{a,a+1,\ldots,b\}$ and we denote $[1,b]$ simply as $[b]$. Further, $\boldsymbol{1}_{A}(w)$ denotes the indicator function of $A$ and equals $1$ if $w\in A$ and $0$ otherwise. 

\section{System Model} 
Now, we briefly outline the system setup including the performance metric normalized delivery time (NDT). \alert{In the BRC of Fig. \ref{fig:HetNet}, $M$ RNs and $K$ UEs request arbitrary files, each file of length $L$ bits, from the set of $N$ popular files $\mathcal{W}=\{W_1,W_2,\ldots,W_N\}$. While the BS has access to the entire file library $\mathcal{W}$, the RNs are able to prefetch only $\mu N L$ bits from $\mathcal{W}$ \emph{before} the file delivery unaware of the actual files being requested. The parameter $\mu$ is commonly referred to as \emph{fractional cache size}. It denotes how much content can be stored at each RN relative to the size of the entire library $\mathcal{W}$. Thus, it ranges from $\mu\in[0,1]$. The prefetching schemes are restricted to arbitrary uncoded symmetric caching strategies in which at most $\mu L$ bits of each file $W_n,n=1,2,\ldots,N$, are cached.} The file that is requested by the $i$-th node is denoted by $W_{d_i}\in\mathcal{W}$. Hereby, $d_i$ represents the demand index of the $i$-th node\footnote{By convention, the first $K$ nodes denote UEs whereas the remaining $M$ nodes the RNs.}. Concatenating the demand indices of UEs and RNs gives the demand vector $\boldsymbol{d}=(d_1,d_2,\ldots,d_{K+M})$. This vector is shared among all nodes prior to the file delivery. 
\begin{figure*}
	\centering
	\begin{subfigure}[b]{0.475\textwidth}
		\centering
		\vspace{8.5em}
		%\hspace{-2em}
		\begin{tikzpicture}[scale=1]
		\SymModSchemePhaseOne
		\end{tikzpicture}
		\vspace{-10em}
		\caption{\small Scheme for the first phase ($t\in[T_1]$)} 
		\label{fig:OS_Phase1}
	\end{subfigure}
	\hfill
	\begin{subfigure}[b]{0.475\textwidth}  
		\centering 
		\vspace{8.5em}
		%\hspace{-2em}
		\begin{tikzpicture}[scale=1]
		\SymModSchemePhaseTwo
		\end{tikzpicture}
		\vspace{-10em}
		\caption{\small Scheme for the second phase ($t\in[T_1+1:T_1+T_2]$)}  
		\label{fig:OS_Phase2}       
	\end{subfigure}
	\caption
	{\small Illustration of the proposed one-shot scheme with $M=4$ RNs, $K=2$ UEs and $\mu M=2$ for the worst-case demand scenario. On the one hand, in each channel use of the first phase [cf. (a)], the MAN scheme is used on the SISO BS-RN broadcast channel to convey desired symbols of (any combination of) $1+\mu M$ RNs. In the worst-case scenario, where UEs request other files, these symbols represent interference which are \emph{partially} zero-forced through cooperative BS-RN beamforming at (any combination) of $\min\{K,\mu M\}$ UEs. Simultaneously, the scheme exploits RN caches by providing the same UEs with \alert{$\alpha\log(P)$ bits of the desired file}. After $T_1$ channel uses of the first phase, the demand of the RNs is satisfied. On the other hand, the second phase [cf. (b)] is devoted to communicate, if necessary, the remaining file symbols of the UEs by applying cooperative BS-RN zero-forcing beamforming.} 
	\label{fig:OS_Scheme}
\end{figure*}
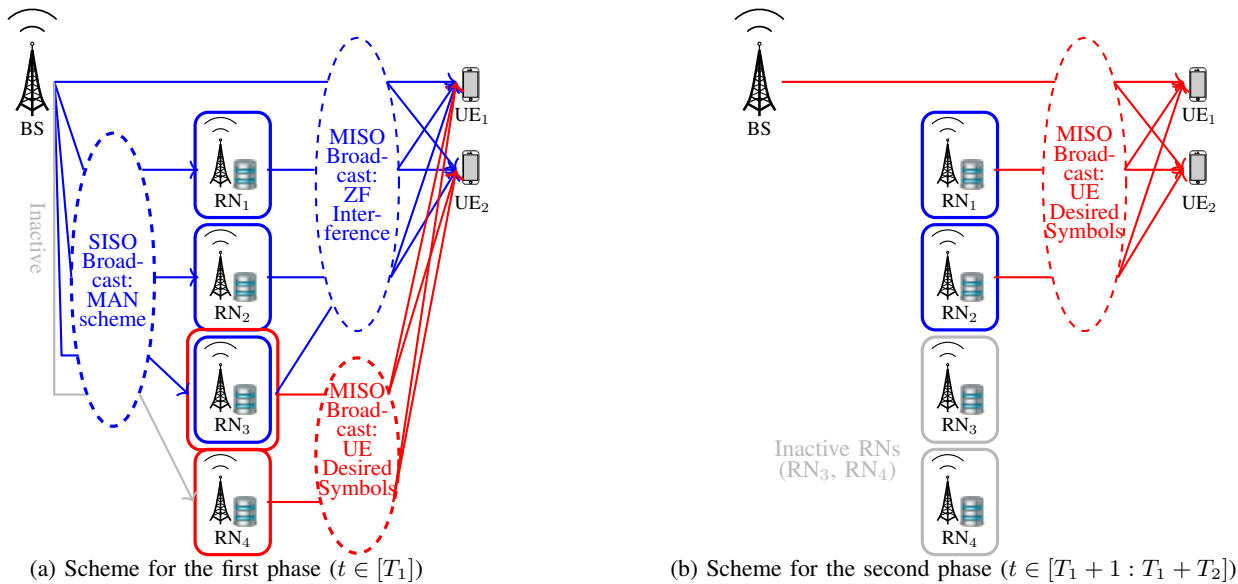 

For $T$ channel uses, the input-output equations of the BRC are given by %(for $k=1,2,\ldots,K$ and $m=1,2,\ldots,M$)
\begin{align}
	\label{eq:rx_sig_UE}
	\boldsymbol{y}_{k}^{(u)}&=g_{k}\boldsymbol{x}^{(s)}+\sum_{m=1}^{M}h_{km}\boldsymbol{x}_{m}^{(r)}+\boldsymbol{z}_{k}^{(u)},\\
	\label{eq:rx_sig_RN}
	\boldsymbol{y}_{m}^{(r)}&=f_{m}\boldsymbol{x}^{(s)}+\boldsymbol{z}_{m}^{(r)},
\end{align} 
for UE$_k$ and RN$_m$, respectively. $\boldsymbol{x}^{(s)},\boldsymbol{x}_{m}^{(r)}\in\mathbb{C}^{T}$ are the transmit signals at the BS and the $m$-th RN, respectively, satisfying average power constraints $\mathbb{E}[|\boldsymbol{x}_{i}^{(s)}|^{2}]\leq P$ and $\mathbb{E}[|\boldsymbol{x}_{m,i}^{(r)}|^{2}]\leq P$, where $\boldsymbol{x}_{i}^{(s)}$ and $\boldsymbol{x}_{m,i}^{(r)}$ are the $i$-th components of $\boldsymbol{x}^{(s)}$ and $\boldsymbol{x}_{m}^{(r)}$. The vectors $\boldsymbol{z}_{k}^{(u)}$ and $\boldsymbol{z}_{m}^{(r)}$ are additive Gaussian noise terms at UE$_k$ and RN$_m$ consisting of i.i.d. components of zero mean and unit variance (denoted $\mathcal{CN}(0,1)$). All $M$ RNs are assumed to be causal and full duplex. Further, in \eqref{eq:rx_sig_UE}-\eqref{eq:rx_sig_RN}, $f_m$ and $g_k$ represent the complex channel coefficients from BS to RN$_m$ and UE$_k$, respectively, while $h_{km}$ is the channel from RN$_m$ to UE$_k$. For notational simplicity, we summarize the channel state information (CSI) by the channel vectors $\boldsymbol{f}=\{f_m\}_{m=1}^{m=M}$, $\boldsymbol{g}=\{g_k\}_{k=1}^{k=K}$ and the matrix $\boldsymbol{H}=\{h_{km}\}_{k=1,m=1}^{k=K,m=M}$. We denote the number of channel uses required to satsify all file demands by $T$. This time is the delivery time which depends on the demand vector $\boldsymbol{d}$ and the \emph{channel estimates} of $\boldsymbol{f},\boldsymbol{g}$ and $\boldsymbol{H}$, i.e., $T=T(\boldsymbol{d},\boldsymbol{\hat{f}},\boldsymbol{\hat{g}},\boldsymbol{\hat{H}})$\footnote{We shall occasionally avoid indexing the functional dependency for notational simplicity.}. \alert{The channel estimates are of relevance because we assume that both the BS and all RNs have only access to imperfect CSI.} In detail, the BS, knows only the imperfect estimates
$\boldsymbol{\hat{f}},\boldsymbol{\hat{g}}$ and $\boldsymbol{\hat{H}}$, while each RN is aware of $\boldsymbol{\hat{g}},\boldsymbol{\hat{H}}$ and $\boldsymbol{f}$\footnote{This assumption is in agreement with the widely used imperfect and perfect CSI setting at transmitting and receiving nodes, respectively. \alert{Thus, the UEs have perfect-quality CSI.}}.
%\begin{itemize}
%	\item $\boldsymbol{\hat{f}},\boldsymbol{\hat{g}}$ and $\boldsymbol{\hat{H}}$,
%\end{itemize} while each RN is aware of
%\begin{itemize}
%	\item $\boldsymbol{\hat{g}},\boldsymbol{\hat{H}}$ and $\boldsymbol{f}$\footnote{This assumption is in agreement with the widely used imperfect and perfect CSI setting at transmitting and receiving nodes, respectively.}.
%\end{itemize} 
Each channel vector and matrix entry can be written as
%\begin{align}
%	f_{m}&=\hat{f}_{m}+\tilde{f}_m,\nonumber\\
%	g_{k}&=\hat{g}_{k}+\tilde{g}_k,\\
%	h_{km}&=\hat{h}_{km}+\tilde{h}_{km}.\nonumber
%\end{align} 
$f_{m}=\hat{f}_{m}+\tilde{f}_m,$ $g_{k}=\hat{g}_{k}+\tilde{g}_k$ and $h_{km}=\hat{h}_{km}+\tilde{h}_{km},$ $\forall m,k$. The estimation errors of each channel are assumed to be of same quality in MSE-sense, i.e., $\mathbb{E}[|\tilde{f}_m|^{2}]=\mathbb{E}[|\tilde{g}_k|^{2}]=\mathbb{E}[|\tilde{h}_{km}|^{2}]=\sigma^{2}(P)<1$. We define the \emph{CSI quality parameter} $\alpha\in[0,1]$ as the power exponent of the estimation error in the high SNR regime as (cf. \cite{KakarMDPI})
\begin{equation}
	\alpha=\lim\limits_{P\rightarrow\infty}-\frac{\log\sigma^{2}}{\log P}.
\end{equation} We observe that $\sigma^{2}$ scales with $P^{-\alpha}$, i.e., $\sigma^{2}\doteq P^{-\alpha}$, where $\doteq$ denotes the exponential equality\footnote{We use this equality in the form $f(P)\doteq P^{c}$ to denote $\lim\limits_{P\rightarrow\infty}\frac{\log f(P)}{\log P}=c$.}. The extreme cases of $\alpha=0$ and $\alpha=1$ represent the channel settings of no CSI and (quasi) perfect CSI, respectively. Now we are ready to define the delivery time per bit and its normalized version.   
\begin{definition}(Delivery time per bit \cite{avik}) 
	The delivery time per bit (DTB) is defined as
		\begin{equation}\label{eq:DTB}
		\Delta(\mu,\sigma^{2},P)=\max_{\boldsymbol{d}\in [N]^{K+M}}\limsup_{L\rightarrow\infty}\frac{\mathbb{E}[T(\boldsymbol{d},\boldsymbol{\hat{f}},\boldsymbol{\hat{g}},\boldsymbol{\hat{H}})]}{L},
		\end{equation} 
		where the expectation is over the channel realizations.
\end{definition} The normalization of the expected delivery time by the file size $L$ gives insight about the per-bit delivery time. In this context, the DTB
measures the time needed per-bit when transmitting the requested
files through the wireless channel for the worst-case request pattern of RNs and UEs as $L\rightarrow\infty$. The ratio of two DTBs -- the DTB of the network under study over the DTB of an interference-free system (e.g., Gaussian point-to-point channel) given by $1/\log(P)$ -- in the high SNR-regime helps us define the NDT.
\begin{definition} (Normalized delivery time \cite{avik}) 
	The NDT is defined as 
	\begin{equation}\label{eq:NDT}
	\delta(\mu,\alpha)=\lim_{P\rightarrow\infty}\frac{\Delta(\mu,\sigma^{2},P)}{1/\log(P)}.
	\end{equation} 
	We denote the minimum NDT by $\delta^{\star}(\mu,\alpha)$.
\end{definition} In short, the NDT $\delta(\mu,\alpha)$ indicates that the worst-case delivery time for one bit of the cache-aided network at fractional cache size $\mu$ and channel quality parameter $\alpha$ is $\delta(\mu,\alpha)$ times larger than the time needed by the reference system.      

\section{General Bounds on the Minimum NDT}

In this section, we provide lower and upper bounds on the NDT for a general BRC that consists of a single BS, $M$ RNs and $K$ UEs. % Specifically, when it comes to the upper bound, we develop an \emph{one-shot} scheme that synergistically exploits both multicasting (coded) caching and distributed zero-forcing opportunities under imperfect CSI at BS and all $M$ RNs. 

\subsection{Lower Bound on the NDT}

The following theorem presents an information-theoretic lower bound on the NDT. 

\begin{theorem}(Lower bound on NDT)\label{theorem_lower_bound}
	For the transceiver cache-aided network with one BS, $M$ RNs each endowed with a cache of fractional cache size $\mu\in[0,1]$, $K$ UEs and a file library of $N\geq K+M$ files, the optimal NDT is lower bounded under \emph{perfect CSI} ($\alpha=1$) at all nodes by
	\begin{align}\label{eq:NDT_lw_bound}
	\delta^{\star}&(\mu,1)\geq\max\Big\{1,\max_{\substack{\ell\in[\bar{s}:M],\\s\in[\min\{M+1,K\}]}}\delta_{\text{LB}}(\mu,\ell,s)\Big\},
	\end{align} where $\bar{s}=M+1-s$ and 
	\begin{align}\label{eq:NDT_lw_bound_inner_comp}
	&\hspace{-.25cm}\delta_{\text{LB}}(\mu,\ell,s)=\frac{K+\ell-\mu(\bar{s}\big(K-s+\frac{(\bar{s}-1)}{2}\big)+\frac{\ell}{2}(\ell+1))}{s}.
	\end{align}
\end{theorem}
\begin{proof}
	The proof is not provided (see \cite{KakarLongVer}) in this paper due to space limitations. %However, we briefly elaborate on the main ideas of the proof. 
	
%	First, we find the bound $\delta_{\text{LB}}(\mu,\ell,s)$ by exploiting the following main observation in the high SNR regime (where noise becomes negligible). That is, given the channel outputs of any $s$ UEs (e.g., of UE$_1$, UE$_2$, $\ldots$, UE$_s$ represented by $\boldsymbol{y}_{k}^{(u)},k\in\{1,2,\ldots,s\}$), in addition to the cached content of $\ell$ RNs (e.g, cached contents of RN$_1$, RN$_2$, $\ldots$, RN$_\ell$) such that $s+\ell\geq M+1$ enables the decoding (with low error probability) of all $K$ files requested by the UEs as well as $\ell$ files desired by the RNs. This is due to the fact that with this information set, all $M+1$ transmit signals consisting of the BS signal $\boldsymbol{x}^{(s)}$ and the RNs transmit signals $\boldsymbol{x}_{m}^{(r)},\forall m\in[M]$, can be reproduced. This in turn, allows the reconstruction of the following channel outputs: On the one hand, the remaining $K-s$ channel outputs of the UEs and on the other hand $\ell$ outputs of the RNs. With the availability of $K$ UE channel ouputs as well $\ell$ RN channel outputs and cached contents, $K+\ell$ files in total become decodable.
%	
%	Second, the unity lower bound follows from the fact that the NDT is bounded from below by the performance of the reference interference-free system with an NDT of $1$. The maximum
%	over these two lower bounds concludes the proof of Theorem \ref{theorem_lower_bound}.
\end{proof}

\subsection{Upper Bound on the NDT}

\begin{figure*}
	%\begin{figure}[h]
	\hspace{3cm}
	\centering
	\begin{minipage}[c]{0.98\linewidth}
		\begin{subfigure}[c]{0.95\textwidth}
			\centering
			\begin{tikzpicture}[scale=0.75]
			\PwrLevelPOne
			\end{tikzpicture}
			\caption{First phase -- Integration of MAN multicasting and partial distributed RN ZF beamforming. \alert{Due to partial zero-forcing, the desired symbol for UE$_k$ is received with full power $P$, whereas the remaining undesired (interfering) UE and RN symbols are received at a reduced power level of $P^{1-\alpha}$.}}
			\label{fig:Pwr_Level_P1}
		\end{subfigure}
	\end{minipage}
	\hspace{1cm}
	%\vfill
	\begin{minipage}[c]{0.95\linewidth}
		\begin{subfigure}[c]{0.9\textwidth}
			\centering
			\begin{tikzpicture}[scale=0.75]
			\PwrLevelPTwo
			\end{tikzpicture}
			\caption{Second phase -- Private and common signaling with distributed BS-RN ZF beamforming of private symbols. \alert{Similarly to phase one, partial ZF reduces the power level of private symbols by $P^{-\alpha}$.}}
			\label{fig:Pwr_Level_P2}
		\end{subfigure}
	\end{minipage}
	\caption
	{\small Power levels of symbols at BS, RN$_m$ and UE$_k$ for the proposed one-shot scheme in (a) the first phase and (b) the second phase. In the first phase, UE$_k$ decodes its desired symbol of rate $\alpha\log(P)$ by \emph{treating residual interference as noise}. On the other hand, in the second phase, UE$_k$ uses \emph{successive decoding} to first decode the common symbol (desired by UE$_{\bar{k}}$ of rate $(1-\alpha)\log(P)$) and then cancel the common message from its received signal to retrieve its desired private symbol of rate $\alpha\log(P)$.} 
	\label{fig:Pwr_Level}
	%\end{figure}
\end{figure*}
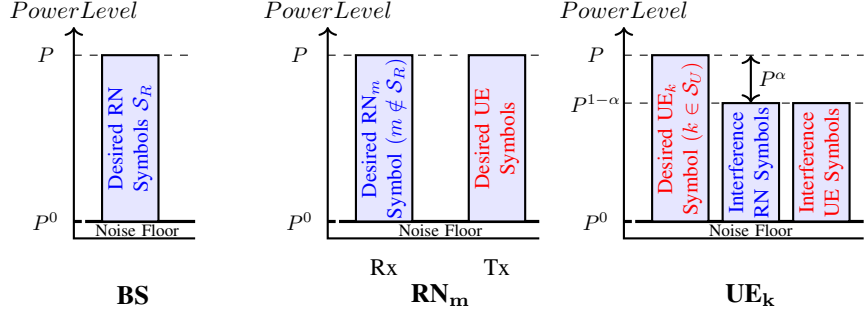
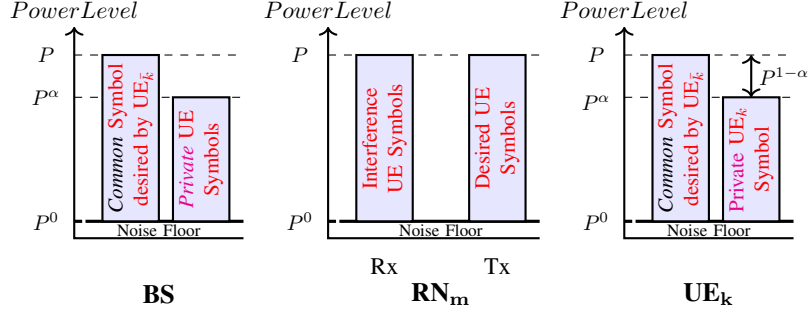
The following theorem specifies the achievable NDT of our proposed one-shot (OS) scheme \alert{which synergistically exploits both multicasting (coded) caching and distributed zero-forcing opportunities under imperfect CSI at BS and all $M$ RNs.} 
\begin{theorem}(Achievable One-Shot NDT)\label{th:one_shot_ach_NDT}
	For $N\geq K+M$ files, $K$ UEs and $M$ RNs each with a cache of (fractional) size $\mu\in\{0,\nicefrac{1}{M},\nicefrac{2}{M},\nicefrac{3}{M},\ldots,\nicefrac{(M-1)}{M},1\}$ and \emph{imperfect} CSI at BS and RN of quality $\alpha$,
		\begin{align}\label{eq:NDT_OS_ach}
			\delta_{\text{OS}}(\mu,\alpha)\triangleq&\max\Bigg\{\delta_{\text{MAN}}(\mu),\nonumber\\
				&\frac{K+\delta_{\text{MAN}}(\mu)\Big(1-\boldsymbol{1}_{K\leq\mu M}(\mu)\alpha\Big)}{1+\min\{K-1,\mu M\}\alpha}\Bigg\}
		\end{align} 
%		\begin{align}\label{eq:NDT_OS_ach}
%			\delta_{\text{OS}}(\mu,\alpha)\triangleq\max\Bigg\{\delta_{\text{MAN}}(\mu),\frac{K+\delta_{\text{MAN}}(\mu)\Big(1-\boldsymbol{1}_{K\leq\mu M}(\mu)\alpha\Big)}{1+\min\{K-1,\mu M\}\alpha}\Bigg\}
%		\end{align} 
	is achievable, where $\delta_{\text{MAN}}(\mu)=M\cdot(1-\mu)\cdot\frac{1}{1+\mu M}$ is the achievable Maddah-Ali Niesen (MAN) NDT			
%		\begin{equation*}
%			\delta_{\text{MAN}}(\mu)=M\cdot(1-\mu)\cdot\frac{1}{1+\mu M} 
%		\end{equation*} 
	such that $\delta^{\star}(\mu)\leq\delta_{\text{OS}}(\mu)$. For arbitrary $\mu\in[0,1]$, the lower convex envelope of these points is achievable.  
\end{theorem}
\begin{proof}
	Herein, we give a sketch of the scheme. We assume that $M$ RNs and $K$ UEs request all distinct files. Recall that the file length is denoted by $L$. Assuming Gaussian signaling, the file is comprised of $L'=L/\log(P)$\footnote{We set $L'$ to $\Gamma{{M}\choose{\mu M}}$ symbols, where $\Gamma={{K}\choose{\psi}}$ and $\psi=\min\{K,\mu M\}$.} symbols. The scheme we develop consists of two phases requiring $T_1$ and $T_2$ channel uses (cf. Fig. \ref{fig:OS_Scheme}), respectively, to send uncached $(1-\mu)L'$ Gaussian symbols (each symbol carrying approximately $\log(P)$ bits) to each RN and also $L'$ symbols to each UE. 
	
	The RNs prefetch their caches at fractional cache sizes $\mu\in\{\nicefrac{1}{M},\nicefrac{2}{M},\ldots,1\}$\footnote{At $\mu=0$, irrespective of $\alpha$, the optimal NDT is $K+M$ achievable by unicasting desired symbols to $K$ UEs and $M$ RNs.} as follows. Any combinations of $\mu M$ RNs\footnote{In total, there are ${{M}\choose{\mu M}}$ combinations} share $\Gamma={{K}\choose{\psi}}$ symbols (per file) each of rate $\log(P)$. In consequence, RN$_m$ caches a total of $\Gamma{{M-1}\choose{\mu M-1}}$ symbols per file which constitutes a fractional cache size of $\mu$. 
	  
	In every channel use of the \emph{first phase} depicted in Fig. \ref{fig:OS_Phase1}, beamforming facilitates the integration of the MAN scheme \cite{Maddah-Ali2} with zero-forcing beamforming to (i) \emph{partially} cancel interference caused by applying the MAN scheme on the BS-RN broadcast channel at the UEs and (ii) convey \alert{$\alpha\log(P)$ bits of the desired file to each of the UEs}. Precisely, the MAN scheme is applied on the BS-RN broadcast channel to provide each RN in a subset $\mathcal{S}_R\subset[M]$ with $|\mathcal{S}_R|=1+\mu M$ RNs with a desired symbol (each of rate $\log(P)$). \emph{Simultaneously}, the full-duplex capabilities at the RNs are exploited by conveying to each UE in the subset $\mathcal{S}_U\subset[K]$ (with $|\mathcal{S}_U|=\min\{K,\mu M\}$ UEs in total) with \alert{$\alpha\log(P)$ bits of its desired file by lowering the residual interference} due to all interfering symbols that $|\mathcal{S}_R|=1+\mu M$ RNs in $\mathcal{S}_R$ desire (cf. Fig. \ref{fig:Pwr_Level_P1} at UE$_k$). Recall that the first phase consumes $T_1$ channel uses. We can show that $T_1=\Gamma{{M}\choose{1+\mu M}}$ suffice in sending each RN$_m$, $\forall m\in[M]$, the remaining $(1-\mu)L'$ symbols of its requested file. Simultaneously, in $T_1$ channel uses each UE$_k$, $\forall k\in[K]$, receives $\tilde{L}={{M}\choose{1+\mu M}}{{K-1}\choose{\psi-1}}\alpha$ symbols of its desired file
	%, with $\tilde{L}$ being proportional to $|\mathcal{S}_U|=\min\{K,\mu M\}$ and $\alpha$. 
	Thus, we may encounter cases where it is either feasible or infeasible to communicate all $L'$ symbols of each requested file to the respective UEs in $T_1$ channel uses ($\tilde{L}\geq L'$ or $\tilde{L}<L'$).   
%\begin{table*}
%	\begin{center}
%		\begin{tabular}{ |l|l|c|c|l|}
%			\hline
%			\multirow{2}{*}{Region Name} & \multirow{2}{*}{Definition} & \multicolumn{2}{c|}{Channel limitation} & \multirow{2}{*}{Achievable NDT} \\\cline{3-4} & & RN side & UE side &  \\ \hline  \multirow{2}{*}{Region A} & $K\leq\mu M<M<\frac{1}{1-2\mu},\mu\leq\frac{1}{2}$ & \multirow{2}{*}{--} & \multirow{2}{*}{\checkmark} & \multirow{2}{*}{$\delta_{\text{OS}}^{(\text{A})}(\mu)=1$} \\\cline{2-2}  & $K\leq\mu M\leq M,M>\frac{1}{1-2\mu},\mu>\frac{1}{2}$ & & & \\ \hline
%			Region B & $K\leq\mu M,\frac{1}{1-2\mu}\leq M,\mu\leq\frac{1}{2}$ & \multirow{2}{*}{\checkmark} & \multirow{2}{*}{--} & \multirow{2}{*}{$\delta_{\text{OS}}^{(\text{B,E})}(\mu)=\delta_{\text{MAN}}(\mu)$} \\ \cline{1-2} Region E & $\mu M <K\leq\mu M\cdot\delta_{\text{MAN}}(\mu)\leq M$ & & & \\ \hline
%			Region C & $\mu M<M<K$ & \multirow{2}{*}{--} & \multirow{2}{*}{\checkmark} & \multirow{2}{*}{$	\delta_{\text{OS}}^{(\text{C,D})}(\mu)=\frac{K+\delta_{\text{MAN}}(\mu)}{1+\mu M}$} \\ \cline{1-2} Region D & $\mu M\cdot\max\Big\{1,\delta_{\text{MAN}}(\mu)\Big\}<K\leq M$ & & & \\ \hline
%		\end{tabular}
%		\captionof{table}{\small Definition of $(\mu,K,M)$ region triplets for $\alpha=1$ and their achievable one-shot NDT. The achievable one-shot NDT in Region A coincides with the lower bound and is thus NDT-optimal.}
%		%\caption{}
%		\label{tab:def_reg_ach_NDT}
%	\end{center}
%\end{table*}	
	Only in the case of missing symbols ($\tilde{L}<L'$) that all $K$ UEs still require after $T_1$ channel uses, additional $T_2>0$ channel uses are required in phase two to deliver the remaining desired symbols as shown in Fig. \ref{fig:OS_Phase2}. To this end, in every channel use private and common signaling in conjuction with \emph{cooperative BS-RN zero-forcing beamforming} of private symbols is deployed to send (i) private symbols (of rate $\alpha\log(P)$) in total to $\psi'=\min\{K,1+\mu M\}$ UEs (ii) and a common symbol (of rate $(1-\alpha)\log(P)$) desired by a single UE (say UE$_{\bar{k}}$ as illustrated in Fig. \ref{fig:Pwr_Level_P2}). 
	%The decoding at the RNs and UEs does not involve symbol decoding over multiple channel uses. Instead, decoding occurs on a one-shot, or single channel use, basis. 
	In consequence, phase two spans $T_2=\frac{K(L'-\tilde{L})}{1+\min\{K-1,\mu M\}\alpha}$ channel uses. In conclusion, the achievable NDT becomes either $\frac{T_1}{L'}=\delta_{\text{MAN}}(\mu)$ if $T_2=0$ or $\frac{T_1+T_2}{L'}=\frac{K+\delta_{\text{MAN}}(\mu)(1-\boldsymbol{1}_{K\leq\mu M}(\mu)\alpha)}{1+\min\{K-1,\mu M\}\alpha}$ if $T_2>0$ with $L'$ being the number of symbols per file.  
\end{proof}

\section{Discussion of Special Cases}

In this section, we investigate the NDT-optimality of special cases of the parameters $\mu$ and $\alpha$. % Further, we define regions with respect to $(K,M,\mu,\alpha)$ for which the functional NDT behavior changes. 

\subsection{Zero-Cache ($\mu=0$)}

In the special case of zero-cache, it is easy to see by comparing the lower bound \eqref{eq:NDT_lw_bound} (for $(\ell,s)=(M,1)$) with the upper bound \eqref{eq:NDT_OS_ach} that the CSI quality does not play any role and the optimal NDT becomes
\begin{equation}
	\delta^{\star}(0,\alpha)=K+M.
\end{equation}   

\subsection{Full-Cache ($\mu=1$)}

On the other hand, in case of full RN caches, the BRC reduces to a MISO BC with $M+1$ antennas and $K$ UEs. The achievable one-shot NDT for this setting corresponds to $\delta_{\text{OS}}(1,\alpha)=\frac{K}{1+\min\{K-1,M\}\alpha}$,
%\begin{equation}
%	\delta_{\text{OS}}(1,\alpha)=\frac{K}{1+\min\{K-1,M\}\alpha},
%\end{equation} 
whereas the lower bound \eqref{eq:NDT_lw_bound} (for $(\ell,s)=(0,M+1)$) gives $\delta^{\star}(1,1)\geq\max\Big\{1,\frac{K}{M+1}\Big\}$.
%\begin{equation}
%	\delta^{\star}(1,1)\geq\max\bigg\{1,\frac{K}{M+1}\bigg\}.
%\end{equation} 
Further, a degrees-of-freedom upper bound on MISO BCs under imperfect CSIT for $M+1$ antennas and $K$ UEs %satisfying $M+1\leq K$ \cite{Davoodi16} 
suggests the following NDT lower bound \alert{for arbitrary $\alpha$ $\delta^{\star}(1,\alpha)\geq\frac{K}{1+\max\{K-1,M\}\alpha}$ \cite{Davoodi16}.}
%\begin{equation}
%	\delta^{\star}(1,\alpha)\geq\frac{K}{1+(K-1)\alpha}.
%\end{equation} 
In conclusion, the proposed one-shot is optimal (i) for $\alpha=0$ at $\mu=1$ for any $K,M$ and (ii) for arbitrary $\alpha$ at fractional cache size $\mu=1$ when $K=M+1$. 

\subsection{Perfect CSI ($\alpha=1$)}

The delivery time of the proposed one-shot scheme is devoted to both RNs and UEs. There are cases where either the BS-RN BC or the BS-RN to UE IC represent the limitation with respect to the delivery time. It is of interest to determine when which limitation happens as a function of $\mu, K$ and $M$. Specifically, when neglecting the discretization of the fractional cache size $\mu$ to values $\{0,\nicefrac{1}{M},\nicefrac{2}{M},\ldots,1\}$, Table \ref{tab:def_reg_ach_NDT} specifies how the one-shot NDT expression \eqref{eq:NDT_OS_ach} simplifies for the given region triplets $(\mu,K,M)$. The regions of Table \ref{tab:def_reg_ach_NDT} are illustrated in Fig. \ref{fig:Region_plot_advanced} for constant $K$ ($K=2$). All discrete points inside Region A lead to the optimal NDT of $\delta^{\star}(\mu,1)=1$. Further, when $M\geq 2K + 1$, we see that for $\mu\geq\nicefrac{1}{M}$, the achievable one-shot
NDT \alert{in regions B and E} does not depend on $K$, i.e., the number of UEs. Instead, the NDT is solely dependent on the number of RNs $M$ and given by the \emph{broadcast-limited} NDT $\delta_{\text{MAN}}(\mu)$. In contrast, for $M\leq 2K$, the IC represents the limitation and the \emph{interference-limited} NDT of $\frac{K+\delta_{\text{MAN}}(\mu)}{1+\mu M}$ is attained.   
\begin{table}
	\begin{center}
		\begin{tabular}{|l|l|l|}
			\hline
			\multirow{1}{*}{Region Name} & \multirow{1}{*}{Definition} & \multirow{1}{*}{Achievable NDT} \\ \hline  \multirow{2}{*}{Region A} & $K\leq\mu M<M<\frac{1}{1-2\mu},\mu\leq\frac{1}{2}$ & \multirow{2}{*}{$1$} \\\cline{2-2}  & $K\leq\mu M\leq M,M>\frac{1}{1-2\mu},\mu>\frac{1}{2}$ & \\ \hline
			Region B & $K\leq\mu M,\frac{1}{1-2\mu}\leq M,\mu\leq\frac{1}{2}$ & \multirow{2}{*}{$\delta_{\text{MAN}}(\mu)$} \\ \cline{1-2} Region E & $\mu M <K\leq\mu M\cdot\delta_{\text{MAN}}(\mu)\leq M$ & \\ \hline
			Region C & $\mu M<M<K$ & \multirow{2}{*}{$	\frac{K+\delta_{\text{MAN}}(\mu)}{1+\mu M}$} \\ \cline{1-2} Region D & $\mu M\cdot\max\Big\{1,\delta_{\text{MAN}}(\mu)\Big\}<K\leq M$ & \\ \hline
		\end{tabular}
		\captionof{table}{\small Definition of $(\mu,K,M)$ region triplets for $\alpha=1$ and their achievable one-shot NDT. The achievable one-shot NDT in Region A coincides with the lower bound and is thus NDT-optimal.}
		%\caption{}
		\label{tab:def_reg_ach_NDT}
	\end{center}
\end{table}
\begin{figure}
	\centering
	%\vspace{-10em}
	\begin{tikzpicture}[scale=0.925]
	%\hspace{15em}
	\PlotMMuRegionsTwo
	\end{tikzpicture}
	\caption[Plot of all regions]{\small 2D $(\mu,M)$-plot of all Regions A, B, C, D and E for constant $K$ $(K=2)$. The labels on the graph indicate the functional relationship at the borders of neighboring regions. The discrete points illustrate the fractional cache sizes $\mu\in\Big\{0,\frac{1}{M},\ldots,\frac{M-1}{M},1\Big\}$ for which the achievable one-shot NDT expression $\delta_{\text{OS}}(\mu)$ in \eqref{eq:NDT_OS_ach} actually hold. The annotations to the regions specify the main characteristics of the respective region. The channel limitations specify which channel -- broadcast or interference channel -- is characteristic for the delivery time overhead. The RN standalone frontier, where $\mu M=K$ holds, represents scenarios for which all $K$ UEs can be served by any subset of $\mu M$ RNs \emph{without} the need of the BS.}
	\label{fig:Region_plot_advanced}
\end{figure}
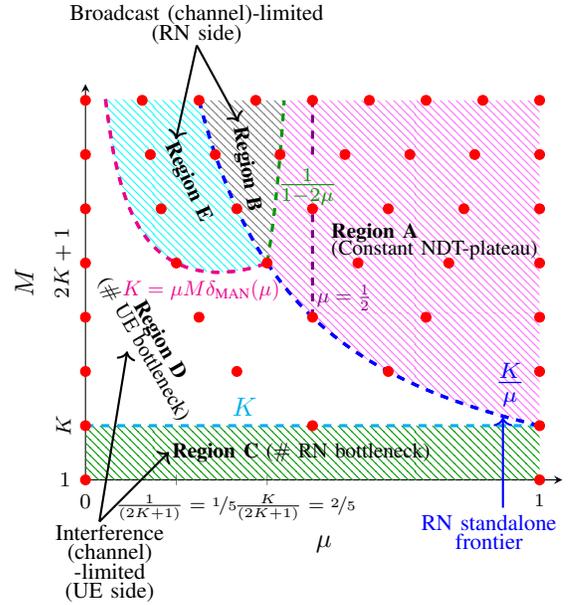 

%\section{Conclusion}
%
%In this paper, we studied the fundamental limits on the delivery time for a cache-aided BRC. We used the normalized delivery time (NDT) as the performance metric and derived lower and upper bounds (in form of an one-shot scheme) on the optimal NDT. We determined cases where our achievable one-shot scheme is NDT-optimal.    

\bibliographystyle{IEEEtran}
\bibliography{bibliography}
\balance
 
\end{document}